\documentclass[a4paper,10pt]{article}
\usepackage[utf8]{inputenc}

\title{Hybrid Copula Estimators}
\author{%
  Johan Segers\\[1ex]%
  \small Universit\'e catholique de Louvain\\%
  \small ISBA, Voie du Roman Pays 20, bte L1.04.01\\%
  \small B-1348 Louvain-la-Neuve, Belgium\\%
  \small johan.segers@uclouvain.be%
}
\date{November 28, 2014}

\usepackage{amsmath}
\usepackage{amsfonts}
\usepackage{amssymb}
\usepackage{amsthm}
\usepackage{natbib}
\usepackage{hyperref}
\usepackage{enumerate}
\usepackage{dsfont}
\usepackage{bm}
\usepackage{paralist}

\newtheorem{condition}{Condition}[section]
\newtheorem{theorem}[condition]{Theorem}
\newtheorem{lemma}[condition]{Lemma}
\newtheorem{corollary}[condition]{Corollary}

\theoremstyle{definition}
\newtheorem{example}[condition]{Example}
\newtheorem{remark}[condition]{Remark}

\numberwithin{equation}{section}

\newcommand{\abs}[1]{\left\lvert{#1}\right\rvert}
\newcommand{\norm}[1]{\left\lVert{#1}\right\rVert}

\newcommand{\vc}[1]{\bm{#1}}

\newcommand{\Cb}{\mathbb{C}}
\newcommand{\Db}{\mathbb{D}}
\newcommand{\Gb}{\mathbb{G}}
\newcommand{\Nb}{\mathbb{N}}

\newcommand{\prob}[1]{\operatorname{P}\left[{#1}\right]}
\newcommand{\expec}[1]{\operatorname{E}\left[{#1}\right]}
\newcommand{\var}[1]{\operatorname{var}\left[{#1}\right]}
\newcommand{\cov}[1]{\operatorname{cov}\left[{#1}\right]}

\newcommand{\reals}{\mathbb{R}}
\newcommand{\DD}{\mathds{D}}
\newcommand{\eps}{\varepsilon}
\newcommand{\dto}{\rightsquigarrow}
\newcommand{\inv}{^\leftarrow}
\newcommand{\1}{\mathds{1}}
\newcommand{\id}{\mathrm{id}}

\renewcommand{\le}{\leqslant}
\renewcommand{\ge}{\geqslant}

\begin{document}

\maketitle

\begin{abstract}
An extension of the empirical copula is considered by combining an estimator of a multivariate cumulative distribution function with estimators of the marginal cumulative distribution functions for marginal estimators that are not necessarily equal to the margins of the joint estimator. Such a hybrid estimator may be reasonable when there is additional information available for some margins in the form of additional data or stronger modelling assumptions. A functional central limit theorem is established and some examples are developed.
\end{abstract}

\section{Introduction}
\label{sec:intro}

Let $H$ be a $p$-variate cumulative distribution function with continuous margins $F_1, \ldots, F_p$ and copula $C$ \citep{sklar:1959}. We have
\begin{align*}
  H( \vc{x} ) &= C \bigl( F_1(x_1), \ldots, F_p(x_p) \bigr), && \vc{x} \in \reals^p, \\
  C( \vc{u} ) &= H \bigl( F_1\inv(u_1), \ldots, F_p\inv(u_p) \bigr), && \vc{u} \in [0, 1]^p.
\end{align*}
Here, $G\inv$ denotes the left-continuous inverse of a univariate cumulative distribution function $G$, i.e.,
\[
  G\inv(u) = \inf \{ x \in \reals : G(x) \ge u \}, \qquad u \in [0, 1].
\]
Throughout, standard conventions regarding infinities are employed: $\inf \varnothing = +\infty$, $G(-\infty) = 0$, and $G(+\infty) = 1$.

Let $\hat{H}_n$ and $\hat{F}_{n,j}$ be estimator sequences of $H$ and $F_j$ ($j = 1, \ldots, p$), respectively. Consider the copula estimator
\begin{equation}
\label{eq:hybrid}
  \hat{C}_n( \vc{u} ) = \hat{H}_n \bigl( \hat{F}_{n,j}\inv(u_1), \ldots, \hat{F}_{n,p}\inv(u_p) \bigr),
  \qquad \vc{u} \in [0, 1]^p.
\end{equation}
Note that $\hat{F}_{n,j}$ is not necessarily equal to the $j$th marginal distribution function, $\hat{H}_{n,j}$, of $\hat{H}_n$. We call $\hat{C}_n$ a \emph{hybrid copula estimator}.

Given a rate $0 < r_n \to \infty$ (typically $r_n = \sqrt{n}$), the normalized estimation error of the hybrid copula estimator is
\begin{equation}
\label{eq:Cbn}
  \Cb_n( \vc{u} ) = r_n \bigl( \hat{C}_n( \vc{u} ) - C( \vc{u} ) \bigr), \qquad \vc{u} \in [0, 1]^p.
\end{equation}
The aim is to establish weak convergence of $\Cb_n$ in the space $\ell^\infty([0, 1]^p)$ of bounded, real-valued functions on $[0, 1]^p$ equipped with the supremum norm.

If $\hat{H}_n$ and $\hat{F}_{n,j} = \hat{H}_{n,j}$ are the joint and marginal empirical distribution functions of a $p$-variate sample of size $n$, then $\hat{C}_n$ is just the Deheuvels--R\"uschendorf empirical copula, see Examples~\ref{ex:empcopproc:iid} and~\ref{ex:empcopproc:dep} below. However, there may be good reasons not to estimate $F_j$ by $\hat{H}_{n,j}$ but by a different estimator. It may be that there is information available on the $j$th margin which cannot directly be used by the joint estimator $\hat{H}_n$.
\begin{itemize}
\item
A parametric model may be reasonable for some or all of the marginal distributions but not for the joint distribution (Example~\ref{ex:parametric}). This is the case for instance when the data are vectors of annual maxima. Asymptotic theory then suggests to model the vector of componentwise maxima by a multivariate max-stable distribution \citep{dehaan:resnick:1977, deheuvels:1978, galambos:1978}. The marginal distributions are univariate extreme-value distributions, whereas the copula belongs to the infinite-dimensional family of extreme-value copulas.
\item
Some entries in the $n \times p$ data matrix may be missing (Example~\ref{ex:missing}). Then $\hat{H}_n$ may be defined as the empirical distribution function of all data rows which are complete, whereas $\hat{F}_{n,j}$ is the empirical distribution function of all observed entries in the $j$th column.
\item
Similarly, in a time series setting, the observation periods of the $p$ univariate series could be different and overlap only partially. Again, one could estimate $F_j$ by the complete series for that variable but estimate $H$ only based on the time period where all series were recorded simultaneously. In the same spirit, there may be additional samples for some of the variables.
\end{itemize}

The structure of the paper is as follows. The main result, Theorem~\ref{thm:hybrid}, is given in Section~\ref{sec:main}, stating weak convergence of the hybrid copula estimator process in \eqref{eq:Cbn} under high-level conditions on the estimators of the joint and marginal distribution functions. Special cases and examples are worked out in Section~\ref{sec:examples}. All proofs and calculations are deferred to Section~\ref{sec:proofs}. 

Throughout, the following notations are used. For an arbitrary set $T$, let $\ell^\infty(T)$ be the space of bounded, real-valued functions on $T$, the space being equipped with the supremum distance $\norm{f}_\infty = \sup_{t \in T} \abs{f(t)}$ for $f \in \ell^\infty(T)$. The indicator variable of a set $E$ is denoted by $\1_E$, whereas the identity mapping on a set $E$ is denoted by $\id_E$. Weak convergence in the sense of J. Hoffmann-J\o{}rgensen is denoted by the arrow `$\dto$'; see Part~1 in the monograph by \cite{vandwell96}.

\section{Main result}
\label{sec:main}

Besides the continuity of the margins $F_1, \ldots, F_p$, two assumptions will be made. The first assumption imposes among others a bit of smoothness on the target copula $C$, without which there is litte hope of establishing weak convergence of $\Cb_n$ in \eqref{eq:Cbn} with respect to the supremum norm on $\ell^\infty([0, 1]^p)$  \citep{segers:2012}. The second assumption is a high-level condition concerning the asymptotic distribution of the estimators $\hat{H}_n$ and $\hat{F}_{n,j}$ and is to be checked on a case-by-case basis. See Remarks~\ref{rem:margins} and~\ref{rem:Donsker} and see the examples in Section~\ref{sec:examples}.

\begin{condition}
\label{cond:C}
\begin{enumerate}[(a)]
\item
The $p$-variate distribution function $H$ has continuous margins $F_1, \ldots, F_p$ and copula $C$.
\item
For all $j \in \{1, \ldots, p\}$, the first-order partial derivative $\dot{C}_j( \vc{u} ) = \partial C( \vc{u} ) / \partial u_j$ exists and is continuous on the set $\{ \vc{u} \in [0, 1]^p : 0 < u_j < 1 \}$.
\end{enumerate}
\end{condition}

For convenience, collect the marginal distribution and quantile functions into vector-valued functions $\vc{F}$ and $\vc{F}\inv$:
\begin{align}
\label{eq:vcF}
  \vc{F}(\vc{x}) &= \bigl( F_1(x_1), \ldots, F_p(x_p) \bigr), && \vc{x} \in \reals^p; \\
\label{eq:vcFinv}
  \vc{F}\inv(\vc{u}) &= \bigl( F_1\inv(u_1), \ldots, F_p\inv(u_p) \bigr), && \vc{u} \in [0, 1]^p.
\end{align}

\begin{condition}
\label{cond:HF}
There exists $0 < r_n \to \infty$ such that in the space $\ell^\infty(\reals^p) \otimes ( \ell^\infty(\reals) \otimes \cdots \otimes \ell^\infty(\reals) )$ equipped with the topology of uniform convergence, we have joint weak convergence
\begin{multline}
\label{eq:HF}
  \bigl( r_n (\hat{H}_n - H); \, r_n (\hat{F}_{n,1} - F_1), \ldots, r_n (\hat{F}_{n,p} - F_p) \bigr) \\
  \dto ( \alpha \circ \vc{F}; \beta_1 \circ F_1, \ldots, \beta_p \circ F_p ), \qquad n \to \infty.
\end{multline}
The stochastic processes $\alpha$ and $\beta_j$ take values in $\ell^\infty([0, 1]^p)$ and $\ell^\infty([0, 1])$, respectively, and are such that $\alpha \circ \vc{F}$ and $\beta_j \circ F_j$ have continuous trajectories on $[-\infty,\infty]^p$ and $[-\infty,\infty]$ almost surely.
\end{condition}

Usually, $r_n = \sqrt{n}$, although Condition~\ref{cond:HF} allows for different convergence rates. Joint weak convergence in \eqref{eq:HF} can typically be established when the estimators $\hat{H}_n$ and $\hat{F}_{n,j}$ can be written as functionals of the same underlying empirical process.
Because $\dot{C}_j( \vc{u} )$ need not be defined if $u_j \in \{0, 1\}$, some care is needed in the formulation of the following theorem.


\begin{theorem}[Hybrid copula process]
\label{thm:hybrid}
If Conditions~\ref{cond:C} and \ref{cond:HF} hold, then, uniformly in $\vc{u} \in [0, 1]^p$, 
\begin{multline}
\label{eq:hybrid:repr}
  r_n \{ \hat{C}_n(\vc{u}) - C(\vc{u}) \} 
  = r_n \{ \hat{H}_n( \vc{F}\inv(\vc{u}) ) - C(\vc{u}) \} \\
  - 
  \sum_{j=1}^p \dot{C}_j(\vc{u}) \, r_n \{ \hat{F}_{n,j}( F_j\inv(u_j)) - u_j \} \, \1_{(0, 1)}(u_j) + o_p(1),
\end{multline}
as $n \to \infty$. Hence, in $\ell^\infty([0, 1]^p)$ equipped with the supremum norm, as $n \to \infty$,
\begin{equation}
\label{eq:hybrid:asym}
  \bigl( r_n \{ \hat{C}_n(\vc{u}) - C(\vc{u}) \} \bigr)_{\vc{u} \in [0, 1]^p}
  \dto 
  \left( \alpha(\vc{u}) - {\textstyle \sum_{j=1}^p} \dot{C}_j(\vc{u}) \, \beta_j(u_j) \right)_{\vc{u} \in [0, 1]^p}.
\end{equation}
The processes $\alpha$ and $\beta_j$ have continuous trajectories almost surely. The right-hand side in \eqref{eq:hybrid:asym} is well-defined because $\beta_j(0) = \beta_j(1) = 0$ almost surely.
\end{theorem}

\begin{remark}[No hybridisation]
\label{rem:margins}
If, as in the standard situation, $\hat{F}_{n,j}$ is equal to the $j$th margin of $\hat{H}_n$ for each $j \in \{1, \ldots, p\}$, then, rather than assuming~\eqref{eq:HF}, it suffices to assume
\begin{equation}
\label{eq:HF:alpha}
  r_n (\hat{H}_n - H) \dto \alpha \circ \vc{F}, \qquad n \to \infty,
\end{equation}
in $\ell^\infty(\reals^p)$, where $\alpha$ is a random element in $\ell^\infty([0, 1]^p)$ with continuous trajectories almost surely. Indeed, by the continuous mapping theorem \citep[Theorem~1.3.6]{vandwell96}, equation~\eqref{eq:HF:alpha} implies equation~\eqref{eq:HF} with 
\[ 
  \beta_j(u_j) = \alpha(1, \ldots, 1, u_j, 1, \ldots, 1), \qquad u_j \in [0, 1], 
\]
with $u_j$ appearing at the $j$th coordinate.
\end{remark}

\begin{remark}[Empirical process representation]
\label{rem:Donsker}
Let $\vc{X}_1, \ldots, \vc{X}_n$ be an independent random sample from $H$. For $f \in L^2(H)$, put
\[
  \Gb_n f = \sqrt{n} \left( \frac{1}{n} \sum_{i=1}^n f(\vc{X}_i) - \expec{ f( \vc{X}_1 ) } \right).
\]
Assume there exists functions $f_{\vc{x}}$ and $f_{x, j}$ in $L^2(H)$ satisfying the following assumptions:
\begin{itemize}
\item
We have, as $n \to \infty$,
\begin{align*}
  \sup_{\vc{x} \in \reals^p} \abs{ \sqrt{n} \{ \hat{H}_n( \vc{x} ) - H( \vc{x} ) \} - \Gb_n f_{\vc{x}} } &= o_p(1), \\
  \sup_{x \in \reals} \abs{ \sqrt{n} \{ \hat{F}_{n,j}(x) - F_j(x) \} - \Gb_n f_{x,j} } &= o_p(1), \qquad j \in \{1, \ldots, p\}.
\end{align*}
\item
We have $f_{\vc{x}} = f_{\vc{x}'}$ in $L^2(H)$ as soon as $F_j(x_j) = F_j(x_j')$ for all $j \in \{1, \ldots, p\}$; similarly $f_{x, j} = f_{x', j}$ in $L^2(H)$ as soon as $F_j(x) = F_j(x')$. 
\item
The maps $\vc{x} \mapsto f_{\vc{x}}$ and $x \mapsto f_{x,j}$ are $L^2(H)$-continuous.
\item 
The collection
\[
  \mathcal{F} = \{ f_{\vc{x}} : \vc{x} \in \reals^p \} \cup \{ f_{x, j} : x \in \reals, 1 \le j \le p \}
\]
is $H$-Donsker, i.e., $\Gb_n \dto \Gb$ as $n\to\infty$ in the space $\ell^\infty( \mathcal{F} )$. The limit $\Gb$ is a tight, centered Gaussian process with covariance function 
\begin{equation}
\label{eq:cov}
  \cov{ \Gb f, \Gb g } = \cov{ f(\vc{X}_1), g(\vc{X}_1) }, \qquad f, g \in \mathcal{F}. 
\end{equation}
\end{itemize}
Then Condition~\ref{cond:HF} is fulfilled with
\begin{align*}
  \alpha( \vc{u} ) &= \Gb f_{\vc{x}(\vc{u})}, & \vc{x}(\vc{u}) &= (F_1\inv(u_1), \ldots, F_p\inv(u_p)), \\
  \beta_j(u) &= \Gb f_{x_j(u),j}, & x_j(u) &= F_j\inv(u).
\end{align*}
It follows that, as $n \to \infty$,
\[
  \sqrt{n} (\hat{C}_n - C)
  \dto \left( \Gb f_{\vc{x}(\vc{u})} - \textstyle{\sum_{j=1}^p} \, \dot{C}_j( \vc{u} ) \, \Gb f_{x_j(u_j)} \right)_{\vc{u} \in [0, 1]^p}.
\]
For each $\vc{u}$, the right-hand side is a zero-mean Gaussian random variable whose variance can be computed via \eqref{eq:cov}, yielding
\begin{multline*}
  \var{ \Gb f_{\vc{x}(\vc{u})} - \textstyle{\sum_{j=1}^p} \, \dot{C}_j( \vc{u} ) \, \Gb f_{x_j(u_j)} } \\
  = \var{f_{\vc{x}(\vc{u})}( \vc{X}_1 ) - \textstyle{\sum_{j=1}^p} \,\dot{C}_j( \vc{u} ) \, f_{x_j(u_j)}(\vc{X}_1) },
  \qquad \vc{u} \in [0, 1]^p.
\end{multline*}

For the usual empirical distribution functions, the above assumptions are fulfilled with $f_{\vc{x}} = \1_{(-\vc{\infty}, \vc{x}]}$ and $f_{x,j} = \1_{\{ \vc{y} : y_j \le x \}}$. The conclusion of Theorem~\ref{thm:hybrid} then leads to the familiar asymptotics for the empirical copula process (Examples~\ref{ex:empcopproc:iid} and~\ref{ex:empcopproc:dep}).
\end{remark}

Let $\DD_\phi$ be the subset of $\DD = \ell^\infty( \reals^p ) \otimes ( \ell^\infty( \reals ) \otimes \ldots \otimes \ell^\infty( \reals ) )$ consisting of all vectors $(H; F_1, \ldots, F_p)$ such that $H$ is a $p$-variate cumulative distribution function and $F_1, \ldots, F_p$ are univariate cumulative distribution functions. Consider the map
\begin{equation}
\label{eq:Phi}
  \phi : \DD_\Phi \to \ell^\infty( [0, 1]^p ) : (H; F_1, \ldots, F_p) \mapsto H \circ \vc{F}\inv,
\end{equation}
with $\vc{F}\inv$ as in \eqref{eq:vcFinv}. One way to show Theorem~2.3 is by an application of the functional delta method \citep[Theorem~3.9.4]{vandwell96} to the map $\phi$, provided the map $\phi$ can be shown to be compact (Hadamard) differentiable. In Section~\ref{sec:proofs}, however, the proof of Theorem~2.3 is based on the extended continuous mapping theorem \citep[Theorem1.11.1]{vandwell96} directly. Since weak convergence of deterministic mappings is equal to ordinary convergence, a by-product of Theorem~2.3 is the compact differentiability of $\phi$. This fact being potentially useful in other contexts too, it is stated explicitly below. Let $\DD_0$ be the subset of $\DD$ consisting of all vectors $h = (\alpha \circ \vc{F}; \beta_1 \circ F_1, \ldots, \beta_p \circ F_p)$ where $\alpha \in \ell^\infty( [0, 1]^p )$ and $\beta_j \in \ell^\infty( [0, 1] )$ are such that $\alpha \circ \vc{F}$ and $\beta_j \circ F_j$ are continuous on $[-\infty, \infty]^p$ and $[-\infty, \infty]$, respectively.  

\begin{corollary}[Compact differentiability]
\label{cor:Hdm}
Let $H$ be a $p$-variate cumulative distribution function with continuous margins $F_1, \ldots, F_p$ and with copula $C$ satisfying Condition~\ref{cond:C}. The map $\phi$ in \eqref{eq:Phi} is Hadamard differentiable at $\theta = (H; F_1, \ldots, F_p)$ tangentially to $\DD_0$. The Hadamard derivative $\phi_\theta'$ evaluated at $h \in \DD_0$ is given by the map $\phi_\theta'(h) \in \ell^\infty( [0, 1]^p )$ defined as
\begin{equation*}
  (\phi_\theta'(h))( \vc{u} ) = \alpha (\vc{u}) - \sum_{j=1}^p \dot{C}_j( \vc{u} ) \, \beta_j(u_j),
  \qquad \vc{u} \in [0, 1]^p.
\end{equation*}
Moreover, $\alpha$ and $\beta_j$ are continuous and $\beta_j(0) = \beta_j(1) = 0$.
\end{corollary}

\section{Special cases and examples}
\label{sec:examples}

\begin{example}[Empirical copula I]
\label{ex:empcopproc:iid}
Let $\vc{X}_i = (X_{i,1}, \ldots, X_{i,p})$, for $i \in \{1, \ldots, n\}$, be an independent random sample from $H$. Let $\hat{H}_n$ and $\hat{F}_{n,j}$ be the joint and marginal empirical distribution functions:
\begin{align*}
  \hat{H}_n( \vc{x} ) &= \frac{1}{n} \sum_{i=1}^n \1(\vc{X}_i \le \vc{x}), && \vc{x} \in \reals^p, \\
  \hat{F}_{n,j}( x_j ) &= \frac{1}{n} \sum_{i=1}^n \1( X_{i,j} \le x_j ), && x_j \in \reals.
\end{align*}
The hybrid copula estimator $\hat{C}_n$ is then equal to the Deheuvels--R\"uschendorf empirical copula \citep{ruschendorf:1976, deheuvels:1979}. By classical empirical process theory (see Remark~\ref{rem:Donsker}), Condition~\ref{cond:HF} is satisfied with $r_n = \sqrt{n}$ and $\alpha$ a $C$-Brownian bridge and $\beta_j( u_j ) = \alpha(1, \ldots, 1, u_j, 1, \ldots, 1)$. Theorem~\ref{thm:hybrid} then just confirms the weak convergence of the empirical copula process \citep{Stute84, FRW04, Tsukahara05, vandervaart:wellner:2007, segers:2012}.
\end{example}

\begin{example}[Empirical copula II]
\label{ex:empcopproc:dep}
Let the random vectors $\vc{X}_1, \ldots, \vc{X}_n$ form a stretch of a stationary time series. By Remark~\ref{rem:margins}, the argument in Example~\ref{ex:empcopproc:iid} remains valid provided weak convergence \eqref{eq:HF:alpha} of the multivariate empirical process holds. The latter is typically true for weakly dependent, strictly stationary time series, in which case $\alpha$ is a centered Gaussian process whose covariance structure also depends on the serial dependence structure of the underlying time series \citep{rio:2000, doukhan:fermanian:lang:2009, dehling:durieu:2011, bucher:volgushev:2013}.
\end{example}

\begin{example}[Known margins]
\label{ex:ideal}
In the hypothetical situation that the margins are known, one may just set $\hat{F}_{n,j} = F_j$ for every $j \in \{1, \ldots, p\}$. Remark~\ref{rem:Donsker} applies with $f_{\vc{x}} = \1_{(-\vc{\infty}, \vc{x}]}$ and $f_{x, j} = 0$. The hybrid copula estimator $\hat{C}_n$ is then equal to the empirical distribution function of the vectors of uniform random variables $(F_1(X_{i,1}), \ldots, F_p(X_{i,p}))$, $i = 1, \ldots, n$. The conclusion is the well-known fact that $\sqrt{n} (\hat{C}_n - C)$ converges to a $C$-Brownian bridge.

In \citet{genest:segers:2010}, this `ideal' hybrid copula estimator was compared to the usual empirical copula. Surprisingly, it was concluded that for many  copulas, the empirical copula actually has the lower asymptotic variance.
\end{example}

\begin{example}[Margins modelled parametrically]
\label{ex:parametric}
Assume that the $j$th margin is modelled by a parametric family $(F_j(\,\cdot\, ;\theta_j) : \theta_j \in \Theta_j)$, where $\Theta_j$ is an open subset of $d_j$-dimensional Euclidean space. Then one may estimate $F_j$ parametrically rather than by the marginal empirical distribution function.

Specifically, let $\vc{X}_1, \ldots, \vc{X}_n$ be a random sample from $H$. Let $\hat{\theta}_{n,j}$ be an estimator of $\theta_j$. Estimate $F_j$ by plugging in the estimator for $\theta_j$:
\[
  \hat{F}_{n,j}(x_j) = F_j(x_j; \hat{\theta}_{n,j}), \qquad x_j \in \reals.
\]
To estimate the joint distribution, take for instance the empirical distribution function
\[
  \hat{H}_n( \vc{x} ) = \frac{1}{n} \sum_{i=1}^n \1( \vc{X}_i \le \vc{x} ), \qquad \vc{x} \in \reals^p.
\]
Combining $\hat{H}_n$ and $\hat{F}_{n,j}$ yields the hybrid copula estimator 
\[
  \hat{C}_n(\vc{u}) = \hat{H}_n \bigl( F_1\inv(u_1; \hat{\theta}_{n,1}), \ldots, F_p\inv(u_p; \hat{\theta}_{n,p}) \bigr),
  \qquad \vc{u} \in [0, 1]^p,
\]
containing both parametric and nonparametric components.

To apply Theorem~\ref{thm:hybrid}, we must check Condition~\ref{cond:HF}. In particular, we need to establish an asymptotic representation for $\hat{F}_{n,j}(x_j)$. Required are some basic smoothness assumption on the parametrization $\theta_j \mapsto F_j(\,\cdot\,; \theta_j)$ together with a central limit theorem for $\hat{\theta}_j$. Specifically, assume the following:
\begin{itemize}
\item[(i)] 
The map $\Theta_j \to \ell^\infty( \reals ) : \theta_j \mapsto F_j(\,\cdot\,; \theta_j)$ is differentiable in the sense that
\begin{multline}
\label{eq:Frechet}
  \sup_{x_j \in \reals} \abs{ F_j(x_j ; \theta_j + h ) - F_j(x_j ; \theta_j) - \textstyle{\sum_{k=1}^{d_j}} h_k \, \dot{F}_{j,k}( x_j ; \theta_j) } \\
  = o( \abs{h} ), \qquad \abs{h} \to 0,
\end{multline}
where $\abs{h}$ is the Euclidean norm of $h \in \reals^{d_j}$ and where $\dot{F}_{j,k}( \, \cdot \, ; \theta_j ) \in \ell^\infty( [-\infty,\infty] )$ is continuous and depends on $x_j$ only through $F_j(x_j; \theta_j)$.

To establish \eqref{eq:Frechet}, check that the partial derivatives of $F_j(x_j; \theta_j)$ with respect to the components of $\theta_j$ exist and are continuous and bounded on compact subsets of $[-\infty, +\infty] \times \Theta_j$.
\item[(ii)] 
The estimator $\hat{\theta}_{n,j}$ admits a linear expansion with influence function $\psi_j = (\psi_{j,1}, \ldots, \psi_{j,d_j})$, i.e.,
\[
  \sqrt{n} (\hat{\theta}_{n,j} - \theta_j) = \frac{1}{\sqrt{n}} \sum_{i=1}^n \psi_j( \vc{X}_i ) + o_p(1), \qquad n \to \infty.
\]
Moreover, $\expec{ \psi_{j,k}(\vc{X}_1) } = 0$ and $\expec{ \psi_{j,k}^2(\vc{X}_1) } < \infty$ for every component $k \in \{1, \ldots, d_j\}$.

The influence function $\psi_j$ may and in general will depend on the unknown value of $\theta_j$. Often, $\psi_j(\vc{x})$ will be a function of $\vc{x}$ only through $x_j$, but this is not required.
\end{itemize}
By the functional delta method \citep[Theorem~3.9.4]{vandwell96}, Assumptions~(i) and~(ii) imply that, as $n \to \infty$,
\begin{align*}
  \sqrt{n} \{ F_j( \, \cdot \, ;\hat{\theta}_{n,j} ) - F_j( \, \cdot \, ; \theta_j ) \}
  &= \sum_{k=1}^{d_j} \sqrt{n} (\hat{\theta}_{n,j,k} - \theta_{j,k}) \, \dot{F}_{j,k}( \, \cdot \, ; \theta_j ) + o_p(1) \\
  &= \sum_{k=1}^{d_j} \frac{1}{\sqrt{n}} \sum_{i=1}^n \psi_{j,k}( \vc{X}_i ) \, \dot{F}_{j,k}( \, \cdot \, ; \theta_j ) + o_p(1),
\end{align*}
the $o_p(1)$ terms referring to remainder terms that converge weakly to zero in the space $\ell^\infty( \reals )$. 
%
%
Remark~\ref{rem:Donsker} applies with $f_{\vc{x}} = \1_{(-\vc{\infty}, \vc{x}]}$ and
\[
  f_{x, j}(\, \cdot \,) = \sum_{k=1}^{d_j} \psi_{j,k}(\, \cdot \,) \, \dot{F}_{j,k}( x ; \theta_j ).
\]
We obtain \eqref{eq:HF} with
\[
  \beta_j(u_j) = \sum_{k=1}^{d_j} \dot{F}_{j,k}( F_j\inv(u_j; \theta_j) ; \theta_j ) \, \Gb \psi_{j,k}.
\]
In view of the conclusion at the end of Example~\ref{ex:ideal}, it is not certain that the hybrid copula estimator performs better than the empirical copula: bringing in the parametric models for the margins in this way is not necessarily helpful. As the above analysis shows, both the parametric models for the margins and the parameter estimators play a role.
\end{example}

\begin{example}[Missing data]
\label{ex:missing}
To show the use of the hybrid copula estimator if some data are missing, consider the following bivariate set-up. Given is an $n \times 2$ data matrix, in each row of which one or both entries may be missing. Formally, the observations consist of a sample of independent, identically distributed quadruples
\[
  (I_i, J_i, I_i X_i, J_i Y_i), \qquad i \in \{ 1, \ldots, n \}.
\]
The indicator variable $I_i$ ($J_i$) is equal to $1$ or $0$ according to whether $X_i$ ($Y_i$) is observed or not. The pairs $(I_i, J_i)$ and $(X_i, Y_i)$ are supposed to be independent, i.e., the data are missing completely at random. The indicators $I_i$ and $J_i$ may be dependent, and the probabilities of observing a data-row partially or completely are $\prob{ I_i = 1 } = p_X > 0$, $\prob{ J_i = 1 } = p_Y > 0$, and $\prob{ I_i = J_i = 1 } = p_{XY} > 0$. The estimation target is the copula, $C$, of the bivariate distribution, $H$, of the pairs $(X_i, Y_i)$. The margins, $F$ and $G$, of $H$ are assumed to be continuous and Condition~\ref{cond:C} is assumed to hold.

The marginal and joint distribution functions may be estimated using the data-rows for which the relevant information is available. For $(x, y) \in \reals^2$, put
\begin{align*}
  \hat{F}_n(x) &= \frac{ \sum_{i = 1}^n \1( X_i \le x, I_i = 1 ) }{ \sum_{i = 1}^n \1( I_i = 1 ) }, \\
  \hat{G}_n(y) &= \frac{ \sum_{i = 1}^n \1( Y_i \le y, J_i = 1 ) }{ \sum_{i = 1}^n \1( J_i = 1 ) }, \\
  \hat{H}_n(x, y) &= \frac{ \sum_{i = 1}^n \1( X_i \le x, Y_i \le y, I_i = J_i = 1 ) }{ \sum_{i = 1}^n \1( I_i = J_i = 1 ) }.
\end{align*}
Condition~\ref{cond:HF} can be verified by embedding the previous estimators in a certain empirical process. The resulting formulas resemble those for the classical empirical copula process, but now the asymptotic variances and covariances are to be multiplied by (the reciprocals of) the observation probabilities $p_X$, $p_Y$ and $p_{XY}$. Details are given at the end of Section~\ref{sec:proofs}.
\end{example}

\section{Proofs}
\label{sec:proofs}

First we show that the processes $\alpha$ and $\beta_j$ in Condition~\ref{cond:HF} are necessarily continuous almost surely. The proof is based on the following lemma.

\begin{lemma}
\label{lem:cnt:aux}
Let $F_1, \ldots, F_p$ be continuous univariate cumulative distribution functions and let $g : [0, 1]^p \to \reals$. If the map $\vc{x} \mapsto g(F_1(x_1), \ldots, F_p(x_p))$ is continuous on $[-\infty, \infty]^p$, then $g$ is continuous on $[0, 1]^p$.
\end{lemma}

\begin{proof}
Let $\vc{u}_n \to \vc{u}$ as $n \to \infty$ in $[0, 1]^p$. We need to show that $g( \vc{u}_n ) \to g( \vc{u} )$ as $n \to \infty$. For any subsequence $N \subset \mathbb{N}$, $\abs{N} = \infty$, we can find a further subsequence $M \subset N$, $\abs{M} = \infty$, along which the following property holds: for all $j \in \{1, \ldots, p\}$ we have either $u_{n,j} \le u_j$ for all $n \in M$ or $u_{n,j} \ge u_j$ for all $n \in M$. It suffices to show that $g( \vc{u}_n ) \to g( \vc{u} )$ as $n \to \infty$ in $M$.

Write $\vc{F}(\vc{x}) = (F_1(x_1), \ldots, F_p(x_p))$. Suppose we can find $\vc{x}_n$ (for $n \in M$) and $\vc{x}$ in $[-\infty,\infty]^p$ such that $\vc{F}(\vc{x}_n) = \vc{u}_n$ and $\vc{F}(\vc{x}) = \vc{u}$ and $\vc{x}_n \to \vc{x}$ as $n \to \infty$ in $M$. By continuity of $g \circ \vc{F}$, we then have
\begin{equation*}
  g( \vc{u}_n ) = g \circ \vc{F}( \vc{x}_n ) \to g \circ{F}( \vc{x} ) = g( \vc{u} ), \qquad n \to \infty, \; n \in M,
\end{equation*}
as required. Hence it suffices to find $(\vc{x}_n)_{n \in M}$ and $\vc{x}$ with the required properties. Fix $j \in \{1, \ldots, p\}$. 
\begin{compactitem}
\item If $u_{n,j} \le u_j$ for all $n \in M$, then define $x_j = \inf \{ y : F_j(y) = u_j \}$ and $x_{n,j} = \sup \{ y : y \le x_j, \; F_j(y) = u_{n,j} \}$.
\item If $u_{n,j} \ge u_j$ for all $n \in M$, then define $x_j = \sup \{ y : F_j(y) = u_j \}$ and $x_{n,j} = \inf \{ y : y \ge x_j, \; F_j(y) = u_{n,j} \}$.
\end{compactitem}
Then $F_j(x_j) = u_j$ and $F_{j}(x_{n,j}) = u_{n,j}$ by continuity of $F_j$. Moreover, $x_{n,j} \to x_j$ as $n \to \infty$ in $M$ by the specific choice of the inverses of $F_j$. Indeed, in the first case, we have, on the one hand, $x_{n,j} \le x_j$ for all $n \in M$ and, on the other hand, $\liminf_n x_{n,j} > F_j\inv(u_j) - \delta = x_j - \delta$ for every $\delta > 0$. The proof in the second case is similar.
\end{proof}

\begin{lemma}
\label{lem:cnt}
With probability one, the trajectories of the processes $\alpha$ and $\beta_1, \ldots, \beta_p$ in Condition~\ref{cond:HF} are continuous.
\end{lemma}

\begin{proof}
This is an immediate consequence of Lemma~\ref{lem:cnt:aux}. For $\beta_j$, apply the lemma with $p = 1$.
\end{proof}

The proof of Theorem~\ref{thm:hybrid} is based on a differentiability property of the map that sends a distribution function to its inverse function. In contrast to Lemma~3.9.20 in \cite{vandwell96}, Lemma~\ref{lem:Vervaat} below does not require the distribution function $F$ to have a density; $F$ need not even be strictly increasing between the two endpoints of its support.

\begin{lemma}
\label{lem:Vervaat}
Let $F_n, F : \reals \to [0, 1]$ be cumulative distribution functions. Assume that $F$ is continuous and assume that there exists a sequence $0 < r_n \to \infty$ and a continuous function $\beta : [0, 1] \to \reals$ such that
\begin{equation}
\label{eq:Vervaat:ass}
  \lim_{n \to \infty} \sup_{x \in \reals} \abs{ r_n \{F_n(x) - F(x)\} - \beta(F(x)) } = 0.
\end{equation}
Then $\beta(0) = \beta(1) = 0$ and
\begin{equation}
\label{eq:Vervaat}
  \lim_{n \to \infty} \sup_{u \in [0, 1]} \abs{ r_n \{F(F_n\inv(u)) - u\} + \beta(u) } = 0.
\end{equation}
In particular,
\begin{equation}
\label{eq:Vervaat:bis}
  \lim_{n \to \infty} \sup_{u \in [0, 1]} \abs{ r_n \{F(F_n\inv(u)) - u\} + r_n \{F_n(F\inv(u)) - u\} } = 0.
\end{equation}
\end{lemma}

An abstract way of stating \eqref{eq:Vervaat} is that the map sending a cumulative distribution function $G$ on $\reals$ to the distribution function $F \circ G\inv$ on $[0, 1]$ is Hadamard differentiable at $F$ tangentially to all functions of the form $\beta \circ F$ for some continuous function $\beta : [0, 1] \to \reals$, the derivative being given by the map $\beta \circ F \mapsto - \beta$.

\begin{proof}
First, note that $\beta(0) = \beta(1) = 0$. Indeed, since $F_n(x) - F(x) \to 0$ as $x \to -\infty$ for each fixed $n$, we can find a sequence $x_n \to -\infty$ sufficiently fast such that $r_n \{ F_n(x_n) - F(x_n) \} \to 0$ as $n \to \infty$ and thus 
\[ 
  \beta(0) 
  = \lim_{n \to \infty} \beta(F(x_n)) 
  = \lim_{n \to \infty} r_n \{F_n(x_n) - F(x_n)\} 
  = 0 
\] 
by uniform convergence. Similarly $\beta(1) = 0$.

It follows that in \eqref{eq:Vervaat}, we can restrict the range in the supremum to $u \in (0, 1]$, since $F(F_n\inv(0)) = F(-\infty) = 0$. [However, $F(F_n\inv(1))$ could be smaller than $1$.] Write
\begin{align*}
  \gamma_n &= r_n (F_n - F), &
  \gamma &= \beta \circ F,
\end{align*}
and note that $F_n = F + r_n^{-1} \gamma_n$. On the one hand, for every $u \in (0, 1]$,
\[
  u \le F_n( F_n\inv(u) )
  = F(F_n\inv(u)) + r_n^{-1} \gamma_n( F_n\inv(u) ),
\]
and thus
\[
  r_n \{ F(F_n\inv(u)) - u \} \ge - \gamma_n( F_n\inv(u) ).
\]
On the other hand, for every $u \in (0, 1]$ and every $\delta > 0$, we have
\begin{align*}
  u &> F_n( F_n\inv(u) - \delta ) \\
  &= F( F_n\inv(u) - \delta ) + r_n^{-1} \gamma_n( F_n\inv(u) - \delta ) \\
  &= F( F_n\inv(u) ) + F( F_n\inv(u) - \delta ) - F( F_n\inv(u) ) + r_n^{-1} \gamma_n( F_n\inv(u) - \delta ),
\end{align*}
and thus
\[
  r_n \{F( F_n\inv(u) ) - u \} < - \gamma_n( F_n\inv(u) - \delta ) + r_n \{ F( F_n\inv(u) ) - F( F_n\inv(u) - \delta ) \}.
\]
Since the latter inequality is true for every $\delta > 0$, we can take the limit as $\delta \to 0$. As $F$ is continuous, we obtain
\[
  r_n \{ F(F_n\inv(u)) - u \} \le - \gamma_n( F_n\inv(u) - )
\]
where $\gamma_n(x-)$ is the left-hand limit of $\gamma_n$ at $x$, a limit which must exist since $\gamma_n$ is the rescaled difference of two cumulative distribution functions. In combination, we find
\begin{equation}
\label{eq:Vervaat:sandwich}
  - \gamma_n( F_n\inv(u) ) \le r_n \{ F(F_n\inv(u)) - u \} \le - \gamma_n( F_n\inv(u) - ), \qquad u \in (0, 1].
\end{equation}
The difference between the left-hand and right-hand sides converges uniformly to zero: indeed, since the sequence $\gamma_n$ converges uniformly to the continuous function $\gamma$, we have
\[
  \lim_{n \to \infty} \sup_{x \in \reals} \abs{ \gamma_n(x-) - \gamma_n(x) } = 0.
\]
To show \eqref{eq:Vervaat}, it then suffices to show that
\[
  \lim_{n \to \infty} \sup_{u \in (0, 1]} \abs{ - \gamma_n( F_n\inv(u) ) + \beta(u) } = 0.
\]
By the triangle inequality and since $\gamma = \beta \circ F$,
\begin{multline*}
  \abs{ - \gamma_n( F_n\inv(u) ) + \beta(u) } \\
  \le
  \abs{ - \gamma_n( F_n\inv(u) ) + \gamma( F_n\inv(u) ) }
  + \abs{ - \beta(F(F_n\inv(u))) + \beta(u) }.
\end{multline*}
The first term on the right-hand side converges to zero uniformly in $u \in (0, 1]$ by uniform convergence of $\gamma_n$ to $\gamma$ on $\reals$. By uniform continuity of $\beta$ on $[0, 1]$, the second term on the right-hand side will converge to zero uniformly in $u \in [0, 1]$ if we can show that 
\[
  \lim_{n \to \infty} \sup_{u \in [0, 1]} \abs{F(F_n\inv(u)) - u} = 0.
\]
But the latter equation is a consequence of \eqref{eq:Vervaat:sandwich}, uniform convergence of $\gamma_n$ to the bounded function $\gamma$, and the fact that $r_n \to \infty$ as $n \to \infty$.

Finally, \eqref{eq:Vervaat:bis} follows from by choosing $x = F\inv(u)$ in \eqref{eq:Vervaat:ass}, yielding
\[
  \lim_{n \to \infty} \sup_{u \in [0, 1]} \abs{ r_n \{ F_n(F\inv(u)) - u \} - \beta(u) } = 0
\]
[note that $F(F\inv(u)) = u$ by continuity of $F$] and then using \eqref{eq:Vervaat} and the triangle inequality.
\end{proof}

\begin{lemma}
\label{lem:Vervaat:random}
Let $F : \reals \to [0, 1]$ be a continuous cumulative distribution function. Let $0 < r_n \to \infty$ and let $\hat{F}_n$ be a sequence of random cumulative distribution functions such that, in $\ell^\infty(\reals)$,
\begin{equation}
\label{eq:Delta:ass}
  r_n ( \hat{F}_n - F ) \dto \beta \circ F, \qquad n \to \infty,
\end{equation}
where $\beta$ is a random element of $\ell^\infty([0, 1])$ with continuous trajectories. Then $\beta(0) = \beta(1) = 0$ almost surely and
\begin{equation}
\label{eq:Delta:op1}
  \sup_{u \in [0, 1]} \abs{ r_n \{ F(\hat{F}_n\inv(u)) - u \} + r_n \{ \hat{F}_n(F\inv(u)) - u \} } = o_p(1).
\end{equation}
As a consequence, in $\ell^\infty([0, 1])$,
\begin{equation}
\label{eq:Delta:d}
  \bigl( r_n \{ F( \hat{F}_n\inv(u) ) - u \} \bigr)_{u \in [0, 1]}
  \dto - \beta, \qquad n \to \infty.
\end{equation}
\end{lemma}

\begin{proof}
First, we show that $\beta(0) = \beta(1) = 0$ almost surely. Define the map $g : \ell^\infty( \reals ) \to \reals$ by $g(\gamma) = \inf_{M > 0} \sup_{x : \abs{x} \ge M} \abs{\gamma(x)} = \limsup_{\abs{x} \to \infty} \abs{\gamma(x)}$. The map $g$ is continuous with respect to the supremum distance. As $\hat{F}_n$ and $F$ are cumulative distribution functions, $g( r_n ( \hat{F}_n - F ) ) = 0$ almost surely. By weak convergence \eqref{eq:Delta:ass} and the continuous mapping theorem \citep[Theorem~1.3.6]{vandwell96}, it follows that $g( \beta \circ F ) = \max \{ \abs{ \beta(0) }, \abs{ \beta(1) } \} = 0$ almost surely too.


Equation~\eqref{eq:Delta:d} follows from combining \eqref{eq:Delta:ass} and \eqref{eq:Delta:op1}; use the triangle inequality and the fact that $u = F(F\inv(u))$.

We will show equation~\eqref{eq:Delta:op1} by an application of the extended continuous mapping theorem \citep[Theorem~1.11.1]{vandwell96}.

Let $\Db_n$ be the collection of all functions $\gamma \in \ell^\infty(\reals)$ such that $F + r_n^{-1} \gamma$ is a cumulative distribution function. In particular, $\gamma(\pm\infty) = \lim_{x \to \pm\infty} \gamma(x) = 0$. Define the map $g_n : \Db_n \to \ell^\infty([0, 1])$ by
\[
  (g_n(\gamma))(u) = r_n \{ F((F + r_n^{-1} \gamma)\inv(u)) - u \} + \gamma(F\inv(u)).
\]
Let $\gamma_n \in \Db_n$ be such that $\gamma_n \to \delta \circ F$ in $\ell^\infty(\reals)$, where $\delta : [0, 1] \to \reals$ is continuous. Put $F_n = F + r_n^{-1} \gamma_n$. Then $\gamma_n = r_n (F_n - F)$ and the conditions of Lemma~\ref{lem:Vervaat} are fulfilled. It follows that, in $\ell^\infty(\reals)$,
\[
  g_n(\gamma_n) 
  = r_n (F \circ F_n\inv - \id_{[0, 1]}) + r_n (F_n \circ F\inv - \id_{[0, 1]})
  \to 0, \qquad n \to \infty,
\]
where `$\id$' refers to the identity mapping. By construction, the maps $\hat{\gamma}_n = r_n( \hat{F}_n - F )$ take values in $\Db_n$. Given the assumption~\eqref{eq:Delta:ass} and the previous limit relation, we can then apply the extended continuous mapping theorem. We find that, in $\ell^\infty([0, 1])$,
\[
  g_n( r_n( \hat{F}_n - F ) ) \dto 0, \qquad n \to \infty.
\]
But this is precisely \eqref{eq:Delta:op1}. 
\end{proof}

\begin{lemma}
\label{lem:Taylor}
Let $C$ be a $p$-variate copula satisfying Condition~\ref{cond:C}(b). Let $0 < r_n \to \infty$ and, for each $n \in \Nb$ and $j \in \{1, \ldots, p\}$, let $\beta_{n,j} \in \ell^\infty([0, 1])$ be such that $0 \le u + r_n \beta_{n,j}(u) \le 1$ for all $u \in [0, 1]$. If, for each $j \in \{1, \ldots, p\}$, we have $\beta_{n,j} \to \beta_j$ in $\ell^\infty([0, 1])$ and if $\beta_j$ is continuous and $\beta_j(0) = \beta_j(1) = 0$, then, uniformly in $\vc{u} \in [0, 1]^p$,
\begin{multline}
\label{eq:Taylor}
  r_n \{ C(u_1 + r_n^{-1} \beta_{n,1}(u_1), \ldots, u_p + r_n^{-1} \beta_{n,p}(u_p)) - C(\vc{u}) \} \\
  = \sum_{j=1}^p \dot{C}_j( \vc{u} ) \, \beta_{n,j}(u_j) \, \1_{(0, 1)}(u_j) + o(1),
  \qquad n \to \infty.
\end{multline}
\end{lemma}

Observe that $\dot{C}_j( \vc{u} )$ is not defined if $u_j \in \{0, 1\}$. This is the reason for including the indicator $\1_{(0, 1)}(u_j)$ on the right-hand side of \eqref{eq:Taylor}.

\begin{proof}
For convenience, write
\[
  \vc{\beta}_{n}( \vc{u} ) = \bigl( \beta_{n,1}(u_1), \ldots, \beta_{n,p}(u_p) \bigr), \qquad \vc{u} \in [0, 1]^p.
\]
Fix $\vc{u} \in [0, 1]^p$ and $n \in \Nb$. Define $f : [0, 1] \to \reals$ by
\[
  f(x) = C( \vc{u} + x \, r_n^{-1} \, \vc{\beta}_n( \vc{u} ) ).
\]
The function $f$ is continuous on $[0, 1]$ and continuously differentiable on $(0, 1)$. Indeed, if $\beta_{n,j}(u_j) \ne 0$, then $u_j$ and $u_j + r_n^{-1} \beta_{n,j}( u_j )$ are two different points in $[0, 1]$, and thus 
\begin{equation}
\label{eq:strictconvex}
  \beta_{n,j}(u_j) \ne 0 \quad \Longrightarrow \quad \forall x \in (0, 1): 0 < u_j + x \, r_n^{-1} \, \beta_{n,j}( u_j ) < 1. 
\end{equation}
The derivative of $f$ is
\begin{equation}
\label{eq:fprime}
  f'(x) = \sum_{j=1}^p \dot{C}_j( \vc{u} + x \, r_n^{-1} \, \vc{\beta}_n(\vc{u}) ) \, r_n^{-1} \, \beta_{n,j}(u_j),
  \qquad x \in (0, 1).
\end{equation}
Because of \eqref{eq:strictconvex}, the right-hand side of \eqref{eq:fprime} is well-defined. By the mean value theorem, there exists $x_n(\vc{u}) \in (0, 1)$ such that
\begin{align*}
  r_n \{ C( \vc{u} + r_n^{-1} \vc{\beta}_n( \vc{u} ) ) - C( \vc{u} ) \}
  &= r_n \{ f(1) - f(0) \} 
  = r_n \, f'( x_n(\vc{u}) ) \\
  &= \sum_{j=1}^p \dot{C}_j( \vc{u} + x_n(\vc{u}) \, r_n^{-1} \, \vc{\beta}_n(\vc{u}) ) \, \beta_{n,j}(u_j).
\end{align*}
By the triangle inequality,
\[
  \abs{ 
    r_n \{ C( \vc{u} + r_n^{-1} \vc{\beta}_n( \vc{u} ) ) - C( \vc{u} ) \}
    -
    \sum_{j=1}^p \dot{C}_j( \vc{u} ) \, \beta_{n,j}(u_j) \, \1_{(0, 1)}(u_j)
  }
  \le \sum_{j=1}^p \Delta_{n,j}(\vc{u})
\]
where, for $\vc{u} \in [0, 1]^p$,
\[
  \Delta_{n,j}(\vc{u}) = 
  \abs{ 
    \dot{C}_j( \vc{u} + x_n(\vc{u}) \, r_n^{-1} \, \vc{\beta}_n(\vc{u}) ) - \dot{C}_j( \vc{u} ) \, \1_{(0, 1)}(u_j) 
  } \, 
  \abs{ \beta_{n,j}(u_j) }.
\]
Fix $j \in \{1, \ldots, p\}$. We need to show that $\lim_{n \to \infty} \norm{ \Delta_{n,j} }_\infty = 0$. Since $\lim_{n \to \infty} \beta_{n,j} = \beta_j$ in $\ell^\infty([0, 1])$, we have $\sup_{n \in \Nb} \norm{ \beta_{n,j} }_\infty = M < \infty$. Fix $\eps > 0$. As $\beta_j(0) = \beta_j(1) = 0$ and $\beta_j$ is continuous, there exists $n(\eps) \in \Nb$ and $\delta(\eps) \in (0, 1/2)$ such that 
\[ 
  \sup \{ \abs{ \beta_{n,j}(u_j) } : n \ge n(\eps), \, u_j \in [0, \delta(\eps)] \cup [1-\delta(\eps), 1] \} \le \eps. 
\]
By increasing $n(\eps)$ if necessary, we can also ensure that $M / r_n \le \delta(\eps) / 2$ for all $n \ge n(\eps)$. Split the supremum of $\Delta_{n,j}( \vc{u} )$ over $\vc{u} \in [0, 1]^p$ into two parts, according to whether $u_j \in [\delta(\eps), 1-\delta(\eps)]$ or not. Write $V_j(\delta) = \{ \vc{u} \in [0, 1]^p : \delta \le u_j \le 1-\delta \}$. 
\begin{itemize}
\item
On the one hand, writing $\abs{ \vc{w} }_\infty = \max \{ \abs{ w_1 }, \ldots, \abs{ w_p } \}$ for $\vc{ w } \in \reals^p$,
\[
  \sup_{\vc{u} \in V_j(\delta(\eps))} \abs{ \Delta_{n,j}( \vc{u} ) }
  \le M \, \sup_{ \substack{\vc{u}, \vc{v} \in V_j(\delta(\eps)/2)\\\abs{ \vc{u} - \vc{v} }_\infty \le M/r_n} } 
  \abs{ \dot{C}_j( \vc{u} ) - \dot{C}_j( \vc{v} ) }.
\]
By uniform continuity of $\dot{C}_j$ on $V_j(\delta)$ for any $\delta > 0$, the right-hand side converges to zero as $n \to \infty$.
\item
On the other hand, for $n \ge n(\eps)$, since $0 \le \dot{C}_j \le 1$,
\[
  \sup_{\vc{u} \in [0, 1]^p \setminus V_j(\delta(\eps))} \abs{ \Delta_{n,j}( \vc{u} ) }
  \le \eps.
\]
\end{itemize}
It follows that 
$
  \limsup_{n \to \infty} \norm{ \Delta_{n,j} }_\infty \le \eps.
$
As $\eps > 0$ was arbitrary, we conclude that $\lim_{n \to \infty} \norm{ \Delta_{n,j} }_\infty = 0$, as required.
\end{proof}

\begin{proof}[Proof of Theorem~\ref{thm:hybrid}]
By Lemma~\ref{lem:Vervaat:random}, we have, in $\ell^\infty([0, 1])$,
\begin{multline}
\label{eq:Vervaat:d:j}
  r_n ( F_j \circ \hat{F}_{n,j}\inv - \id_{[0, 1]} ) \\
  = - r_n ( \hat{F}_{n,j} \circ F_j\inv - \id_{[0, 1]} ) + o_p(1) 
  \dto - \beta_j, \qquad n \to \infty.
\end{multline}
Moreover, $\beta_j(0) = \beta_j(1) = 0$ almost surely.

For notational convenience, consider the random vector
\[
  \hat{\vc{F}}_n\inv( \vc{u} )
  = \bigl( \hat{F}_{n,j}\inv(u_1), \ldots, \hat{F}_{n,p}\inv(u_p) \bigr),
  \qquad \vc{u} \in [0, 1]^p.
\]
The following decomposition is fundamental to the analysis of the hybrid copula estimator $\hat{C}_n = \hat{H}_n \circ \hat{\vc{F}}_n\inv$:
\begin{equation}
\label{eq:hybrid:decomp}
  r_n (\hat{C}_n - C)
  = r_n ( \hat{H}_n \circ \hat{\vc{F}}_n\inv - H \circ \hat{\vc{F}}_n\inv )
  + r_n ( H \circ \hat{\vc{F}}_n\inv - C ).
\end{equation}
We will treat both terms on the right-hand side of \eqref{eq:hybrid:decomp} in turn.

As $H = C \circ \vc{F}$, the first term on the right-hand side in \eqref{eq:hybrid:decomp} is
\begin{multline}
\label{eq:hybrid:decomp:1}
  r_n ( \hat{H}_n \circ \hat{\vc{F}}_n\inv - H \circ \hat{\vc{F}}_n\inv )
  = r_n ( \hat{H}_n \circ \hat{\vc{F}}_n\inv - \hat{H}_n \circ \vc{F}\inv \circ \vc{F} \circ \hat{\vc{F}}_n\inv ) \\
  + r_n ( \hat{H}_n \circ \vc{F}\inv \circ \vc{F} \circ \hat{\vc{F}}_n\inv - C \circ \vc{F} \circ \hat{\vc{F}}_n\inv ).
\end{multline}
\begin{itemize}
\item 
The first term on the right-hand of \eqref{eq:hybrid:decomp:1} is $o_p(1)$ in $\ell^\infty([0, 1]^p)$ provided we can show that
\[
  r_n ( \hat{H}_n - \hat{H}_n \circ \vc{F}\inv \circ \vc{F} ) = o_p(1), \qquad n \to \infty.
\]
But the latter holds in view of the identity $H = H \circ \vc{F}\inv \circ \vc{F}$ (the margins of $H$ are $F_1, \ldots, F_p$ and these are continuous), Condition~\ref{cond:HF}, and the identity $\vc{F} \circ \vc{F}\inv \circ \vc{F} = \vc{F}$. 
\item
Regarding the second term on the right-hand side of \eqref{eq:hybrid:decomp:1}, note that, by \eqref{eq:Vervaat:d:j}, for every $j \in \{1, \ldots, p\}$,
\[
  F_j \circ \hat{F}_{n,j}\inv \dto \id_{[0, 1]}, \qquad n \to \infty,
\]
in $\ell^\infty([0,1])$. Moreover, by Condition~\ref{cond:HF} and the identities $C = H \circ \vc{F}\inv$ and $\vc{F} \circ \vc{F}\inv = \id_{[0, 1]^p}$, we have, in $\ell^\infty([0, 1]^p)$,
\[
  r_n( \hat{H}_n \circ \vc{F}\inv - C ) \dto \alpha \circ \vc{F} \circ \vc{F}\inv = \alpha, \qquad n \to \infty.
\]
By asymptotic uniform equicontinuity \citep[Theorem~1.5.7 and Addendum~1.5.8]{vandwell96}, as $n \to \infty$,
\[
  r_n ( \hat{H}_n \circ \vc{F}\inv \circ \vc{F} \circ \hat{\vc{F}}_n\inv - C \circ \vc{F} \circ \hat{\vc{F}}_n\inv )
  = r_n( \hat{H}_n \circ \vc{F}\inv - C ) + o_p(1)
\]
\end{itemize}
We find that, in $\ell^\infty([0, 1]^p)$,
\begin{equation}
\label{eq:hybrid:a}
  r_n ( \hat{H}_n \circ \hat{\vc{F}}_n\inv - H \circ \hat{\vc{F}}_n\inv )
  = r_n( \hat{H}_n \circ \vc{F}\inv - C ) + o_p(1)
  \dto \alpha, \qquad n \to \infty.
\end{equation}

The second term on the right-hand side in \eqref{eq:hybrid:decomp} is $r_n ( C \circ \vc{F} \circ \hat{\vc{F}}_n\inv - C )$. 
For $n \in \Nb$, let $\Db_n$ be the collection of $p$-tuples $\vc{\gamma} = (\gamma_1, \ldots, \gamma_p) \in \ell^\infty( \reals ) \otimes \cdots \otimes \ell^\infty( \reals )$ such that map $x \mapsto F_j(x) + r_n^{-1} \gamma_j(x)$ is a cumulative distribution function for each $j \in \{1, \ldots, p\}$. Define the map $g_n : \Db_n \to \ell^\infty([0, 1]^p)$ by
\begin{multline*}
  (g_n(\vc{\gamma}))(\vc{u})
  = r_n \{ C \circ \vc{F} \circ (\vc{F} + r_n^{-1} \vc{\gamma})\inv( \vc{u} ) - C( \vc{u} ) \} \\
  - \sum_{j=1}^p \dot{C}_j(\vc{u}) \, r_n \{ F_j \circ (F_j + r_n^{-1} \gamma_j)\inv(u_j) - u_j \} \, \1_{(0, 1)}(u_j),
  \qquad \vc{u} \in [0, 1]^p.
\end{multline*}
Let $\vc{\gamma}_n \in \Db_n$ be such that $\lim_{n \to \infty} \vc{\gamma}_n = \vc{\gamma}$ where $\gamma_j = \beta_j \circ F$ and $\beta_j \in \ell^\infty([0, 1])$ is continuous and satisfies $\beta_j(0) = \beta_j(1) = 1$ for every $j \in \{1, \ldots, p\}$. By Lemma~\ref{lem:Vervaat} with $F_{n,j} = F_j + r_n^{-1} \gamma_{n,j}$, we then have
\[
  r_n \{ F_j \circ (F_j + r_n^{-1} \gamma_j)\inv - \id_{[0, 1]} \} \, \1_{(0, 1)}
  \to - \beta_j, \qquad n \to \infty.
\]
By Lemma~\ref{lem:Taylor}, it then follows that
\[
  g_n(\vc{\gamma}_n) \to 0, \qquad n \to \infty.
\]
By the extended continuous mapping theorem \citep[Theorem~1.11.1]{vandwell96}, it follows that
\[
  g_n( r_n ( \hat{\vc{F}}_n - \vc{F} ) ) \dto 0, \qquad n \to \infty.
\]
But this says exactly that, uniformly in $\vc{u} \in [0, 1]^p$, as $n \to \infty$,
\begin{multline*}
  r_n \{ C \circ \vc{F} \circ \hat{\vc{F}}_n\inv( \vc{u} ) - C( \vc{u} ) \} \\
  = \sum_{j=1}^p \dot{C}_j( \vc{u} ) \, r_n \{ F_j \circ \hat{F}_{n,j}\inv(u_j) - u_j \} \, \1_{(0, 1)}(u_j) + o_p(1).
\end{multline*}
Insert \eqref{eq:Vervaat:d:j} to deduce that, uniformly in $\vc{u} \in [0, 1]^p$ and as $n \to \infty$,
\begin{multline}
\label{eq:hybrid:b}
  r_n \{ C \circ \vc{F} \circ \hat{\vc{F}}_n\inv( \vc{u} ) - C( \vc{u} ) \} \\
  = - \sum_{j=1}^p \dot{C}_j( \vc{u} ) \, r_n \{ \hat{F}_{n,j} \circ F_j\inv(u_j) - u_j \} \, \1_{(0, 1)}(u_j) + o_p(1).
\end{multline}

Collect the representations in \eqref{eq:hybrid:a} and \eqref{eq:hybrid:b} of the two terms on the right-hand side of \eqref{eq:hybrid:decomp} and apply Condition~\ref{cond:HF} to arrive at the stated conclusion.
\end{proof}

\begin{proof}[Details for Example~\ref{ex:missing}]
Consider the following functions from $\{0, 1\}^2 \times \reals^2$ into $\reals$: for $(x, y) \in \reals^2$,
\begin{align*}
  f_1(I, J, X, Y) &= \1(I = 1), &
  g_{1,x}(I, J, X, Y) &= \1(X \le x, I = 1), \\
  f_2(I, J, X, Y) &= \1(J = 1), &
  g_{2,y}(I, J, X, Y) &= \1(Y \le y, J = 1), \\
  f_3 &= f_1 f_2, &
  g_{3,x,y} &= g_{1,x} g_{2,y}.
\end{align*}
Let $P$ denote the common distribution of the quadruples $(I_i, J_i, X_i, Y_i)$. The collection of functions
\[
  \mathcal{F} = \{ f_1, f_2, f_3 \} \cup \{ g_{1,x} : x \in \reals \} \cup \{ g_{2,y} : y \in \reals \} \cup \{ g_{3,x,y} : (x, y) \in \reals^2 \}
\]
is a finite union of VC-classes and thus $P$-Donsker \citep[Chapter~2.6]{vandwell96}. The empirical process $\Gb_n$ defined by
\[
  \Gb_n(f) = \sqrt{n} \left( \frac{1}{n} \sum_{i=1}^n f(I_i, J_i, X_i, Y_i) - \expec{ f(I_1, J_1, X_1, Y_1) } \right), \qquad f \in \mathcal{F},
\]
converges in $\ell^\infty( \mathcal{F} )$ to a $P$-Brownian bridge $\Gb$. For $(x, y) \in \reals^2$,
\begin{align*}
  \hat{F}_n(x) &= \frac{p_X \, F(x) + n^{-1/2} \Gb_n g_{1,x} }{ p_X + n^{-1/2} \Gb_n f_1 }, \\
  \hat{G}_n(y) &= \frac{p_Y \, G(y) + n^{-1/2} \Gb_n g_{2,y} }{ p_y + n^{-1/2} \Gb_n f_2 }, \\
  \hat{H}_n(x, y) &= \frac{p_{XY} \, H(x, y) + n^{-1/2} \Gb_n g_{3,x,y}}{ p_{XY} + n^{-1/2} \Gb f_3 }.
\end{align*}
It follows that, as $n \to \infty$ and uniformly in $(x, y) \in \reals^2$,
\begin{align*}
  \sqrt{n} \{ \hat{F}_n(x) - F(x) \}
  &= p_X^{-1} \Gb_n (g_{1,x} - F(x) \, f_1) + O_p(n^{-1/2}), \\
  \sqrt{n} \{ \hat{G}_n(x) - G(x) \}
  &= p_Y^{-1} \Gb_n (g_{2,x} - G(x) \, f_2) + O_p(n^{-1/2}), \\
  \sqrt{n} \{ \hat{H}_n(x, y) - H(x, y) \}
  &= p_{XY}^{-1} \Gb_n (g_{3,x,y} - H(x, y) f_3) + O_p(n^{-1/2}).
\end{align*}
As a consequence, Condition~\ref{cond:HF} is fulfilled with, for $(u, v) \in [0, 1]^2$,
\begin{align*}
  \beta_1(u) &= p_X^{-1} \Gb(g_{1,F\inv(u)} - u \, f_1), \\
  \beta_2(v) &= p_Y^{-1} \Gb(g_{2,G\inv(v)} - v \, f_2), \\
  \alpha(u, v) &= p_{XY}^{-1} \Gb(g_{3,F\inv(u),G\inv(v)} - C(u, v) \, f_3).
\end{align*}
From these formulas, the variances and covariances can be easily computed: for $(u, u_1, u_2, v, v_1, v_2) \in [0, 1]^6$,
\begin{align*}
  \cov{ \beta_1(u_1), \beta_1(u_2) } &= p_X^{-1} \, \{ u_1 \wedge u_2 - u_1 u_2 \}, \\
  \cov{ \beta_2(v_1), \beta_2(v_2) } &= p_Y^{-1} \, \{ v_1 \wedge v_2 - v_1 v_2 \}, \\
  \cov{ \beta_1(u), \beta_2(v) } &= \frac{p_{XY}}{p_X p_Y} \, \{ C(u, v) - uv \},
\end{align*}
and
\begin{align*}
  \cov{ \alpha(u_1, v_1), \alpha(u_2, v_2) } &= p_{XY}^{-1} \, \{ C( u_1 \wedge u_2, v_1 \wedge v_2 ) - C(u_1, v_1) \, C(u_2, v_2) \}, \\
  \cov{ \alpha(u_1, v), \beta_1(u_2) } &= p_X^{-1} \, \{ C( u_1 \wedge u_2, v ) - C(u_1, v) \, u_2 \}, \\
  \cov{ \alpha(u, v_1), \beta_2(v_2) } &= p_Y^{-1} \, \{ C( u, v_1 \wedge v_2 ) - C(u, v_1) \, v_2 \}.
\end{align*}
\end{proof}

\begin{proof}[Proof of Corollary~\ref{cor:Hdm}]
Fix $h \in \DD_0$. Let $0 < t_n \to 0$ in $\reals$ and let $h_n = ( \gamma_n; \delta_{n1}, \ldots, \delta_{np} ) \in \DD$ be such that $h_n \to h$ in $\DD$ as $n \to \infty$ and $\theta + t_n h_n \in \DD_\phi$ for all $n$. Write $H_n = H + t_n \gamma_n$ and $F_{nj} = F_j + t_n \delta_{nj}$. Then $\gamma_n = r_n (H_n - H)$ and $\delta_{nj} = r_n (F_{nj} - F_j)$ with $r_n = t_n^{-1}$. Consider $h_n$ as a deterministic random element taking values in $\DD$. Then Conditions~\ref{cond:C} and~\ref{cond:HF} are fulfilled with $\hat{H}_n$ and $\hat{F}_{nj}$ replaced by $H_n$ and $F_{nj}$, respectively. By Theorem~\ref{thm:hybrid} and in particular by equation~\eqref{eq:hybrid:asym}, we find, in $\ell^\infty([0, 1]^p)$,
\begin{align*}
  t_n^{-1} \bigl( \phi( \theta + t_n h_n ) - \phi( \theta ) \bigr)
  &= r_n ( H_n \circ \vc{F}_n\inv - C ) \\
  &\dto \phi_\theta'(h), \qquad n \to \infty.
\end{align*}
Since weak convergence of constant maps in a metric space is equivalent to the ordinary convergence of their images, the proof is complete.
\end{proof}

\section*{Acknowledgments}

Fruitful discussions with Christian Genest (McGill University) and Axel B\"ucher (Ruhr-Universit\"at Bochum) are gratefully acknowledged. The research project was funded by contract ``Projet d'Act\-ions de Re\-cher\-che Concert\'ees'' No.\ 12/17-045 of the ``Communaut\'e fran\c{c}aise de Belgique'' and by IAP research network Grant P7/06 of the Belgian government (Belgian Science Policy).

\bibliographystyle{chicago}
\bibliography{Biblio}
\end{document}